\documentclass[11pt]{article}
\usepackage{color}
\usepackage{epsfig}
\usepackage{subfig}
\usepackage{mathrsfs}
\usepackage{booktabs}
\usepackage{graphicx}
\usepackage{epstopdf}
\usepackage{multirow}
\usepackage{amsmath,amssymb}
\usepackage{amsthm}
\usepackage{hyperref}
\usepackage{fancyhdr}
\usepackage{tabularx}
\usepackage{xfrac}

\newtheorem{prop}{Proposition}[section]

\newtheorem{remark}{Remark}[section]

\date{}

\begin{document}

\title{Bridging the Gap between Individuality and Joint Improvisation in the Mirror Game}
\author{Chao Zhai, Michael Z. Q. Chen, Francesco Alderisio, Alexei Yu. Uteshev and Mario di Bernardo \thanks{Corresponding author: Michael Z. Q. Chen (mzqchen@outlook.com). Chao Zhai and Michael Z. Q. Chen are with the Department of Mechanical Engineering, University of Hong Kong, Haking Wong Building, Pok Fu Lam Road, Hong Kong. Alexei Yu. Uteshev is with the Faculty of Applied Mathematics, St. Petersburg State University, Universitetskij pr.35, Petrodvorets, 198504, St. Petersburg, Russia. Francesco Alderisio and Mario di Bernardo are with the Department of Engineering Mathematics, University of Bristol, Merchant Venturers' Building, Woodland Road, Bristol, BS8 1UB, United Kingdom. Mario di Bernardo is also with the Department of Electrical Engineering and Information Technology, University of Naples Federico II, 80125 Naples, Italy.}}

\maketitle

\begin{abstract}

Extensive experiments in Human Movement Science suggest that solo motions are characterized by unique features that define the individuality or motor signature of people. While interacting with others, humans tend to spontaneously coordinate their movement and unconsciously give rise to joint improvisation. However, it has yet to be shed light on the relationship between individuality and joint improvisation. By means of an ad-hoc virtual agent, in this work we uncover the internal mechanisms of the transition from solo to joint improvised motion in the mirror game, a simple yet effective paradigm for studying interpersonal human coordination.
According to the analysis of experimental data, normalized segments of velocity in solo motion are regarded as individual motor signature, and the existence of velocity segments possessing a prescribed signature is theoretically guaranteed.
In this work, we first develop a systematic approach based on velocity segments to generate \emph{in-silico} trajectories of a given human participant playing solo.
Then we present an online algorithm for the virtual player to produce joint improvised motion with another agent while exhibiting some desired kinematic characteristics, and to account for movement coordination and mutual adaptation during joint action tasks.
Finally, we demonstrate that the proposed approach succeeds in revealing the kinematic features transition from solo to joint improvised motions, thus revealing the existence of a tight relationship between individuality and joint improvisation.

\end{abstract}

\section{Introduction}

People suffering from social deficiencies (i.e., schizophrenia or autism) find it hard to engage in social activities and interact with others, which inevitably brings sorrow to themselves and their relatives \cite{boraston07,couture06}. The theory of similarity in Social Psychology suggests that individuals prefer to cooperate with others sharing similar morphological and behavioral features, and that they tend to unconsciously coordinate their movements \cite{fol82,schmidt15,walton15}. It has been shown that motor processes caused by interpersonal coordination are closely related to mental connectedness, and that motor coordination between two people contributes to social attachment \cite{wil09,raff15}.


The \emph{mirror game} provides a simple paradigm to study social interactions and the onset of motor coordination among human beings, as it happens in improvisation theater, group dance and parade marching \cite{noy11,noyfront}. In order to enhance social interaction through motor coordination, it would be desirable to create a virtual player (VP) or computer avatar capable of playing the mirror game with a human subject (typically the patient) either by mimicking similar kinematic characteristics or producing dissimilar ones \cite{chao_smc}. Indeed, this allows to modulate the kinematic similarity of the VP while maintaining a certain level of coordination with the human player (HP) so that the s/he is unconsciously guided towards the direction of some desired movement features.


Motor coordination between two or more effectors in biological systems emerges as a result of the integration of several body parts and functions. Such coordination occurs through two types of control actions: feedback and feed-forward \cite{jor99}. The motor system is able to correct the deviation from the desired movement by means of feedback control, whilst feed-forward control allows it to reconcile the interdependency of the involved effectors and preplan the response to the incoming sensory information, without taking into account how the system reacts to the command signal \cite{des00}.
Inspired by the above motor process of the human body, a computational approach based on optimal control has been proposed in the literature for the VP to interact with other participants and reconcile movement coordination with its own prescribed kinematic features \cite{chao_cdc15,chao_ji15}.


The main challenge is to develop a mathematical model capable of driving the VP to joint-improvise with a HP in the mirror game, while guaranteeing an assigned \emph{motor signature} as defined in \cite{piotr15}.
The first step towards this goal is to design a computational architecture able to generate \emph{in-silico} trajectories reproducing the motor signature exhibited by a certain HP playing solo. In so doing, we propose an approach based on velocity segments \cite{noy14}. The second step is to provide such architecture with an online algorithm allowing the virtual player to produce joint improvised motions and interact with a HP or another VP.
Much research effort has been spent on the design of control architectures for the virtual agent or robot \cite{noy11,chao_cdc15,li16,ata16,hirche14,dumas14,chao_cdc14,kelso_plos09}, but only pre-recorded time series of human players in solo trials have been used to generate the joint motion of a customized VP \cite{chao_mg15}, which limits its movement diversity due to the finite number of available pre-recorded trajectories. The approach we propose here overcomes this drawback by allowing the VP to autonomously exhibit any motor signature with specified kinematic features (characterizing the solo motion of a given HP) during the interaction with another agent.


The outline of this paper is given as follows. In Section \ref{sec:prob} we introduce the experimental paradigm of the mirror game, a quantitative marker of motor signatures, and their construction method. In Section \ref{sec:ca} we focus on the design of a computational architecture for the VP. Specifically, we develop an algorithm capable of generating solo motions with prescribed kinematic features, followed by an online algorithm allowing the VP to produce joint improvised motion with another agent. Experimental validations is carried out in Section \ref{sec:exp} to test the proposed approach. Finally, in Section \ref{sec:con} we draw conclusions and discuss future directions.

\section{Preliminaries}\label{sec:prob}

\subsection{Mirror game}

The mirror game is a simple yet effective paradigm to investigate the onset of social motor coordination between two players and describe their movement imitation at high temporal and spatial resolution \cite{noy11,noy14,noy16}. Figure~\ref{mg} shows the experimental set-up at the University of Montpellier, France.

\begin{figure}
\scalebox{0.75}[0.75]{\includegraphics{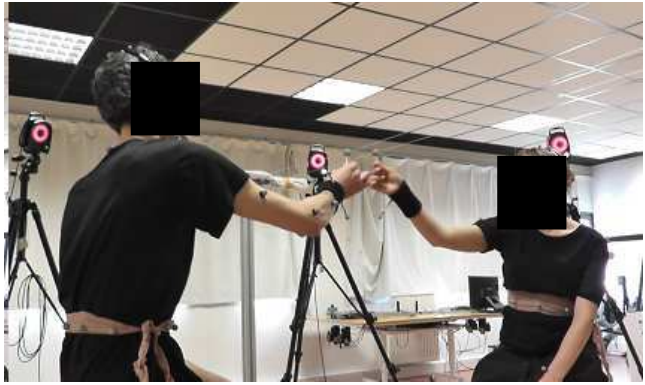}}\centering
\caption{\label{mg} Mirror game set-up at the University of Montpellier \cite{chao_ji15}. Two horizontal strings are mounted perpendicularly at eye level and centrally between the two human participants. Two small balls are mounted on the parallel strings, respectively. Human participants are instructed to hold the handle beneath each ball and move it along the string back and forth. Cameras are installed around the participants to collect experimental data and record their movement trajectories. In solo trials, only one human participant is instructed to perform the motion. In joint trials, two human participants are seated opposite each other and interact while moving their respective ball.}
\end{figure}

The mirror game can be played in three different experimental conditions \cite{piotr15}:
\begin{enumerate}
  \item Solo Condition: This is an individual trial. Participants perform the game on their own and try to create interesting motions.
  \item Leader-Follower Condition: This is a collaborative round, whose purpose is for the participants to create synchronized motions. One player leads the game, while the other tries to follow the leader's movement.
  \item Joint-Improvisation Condition: Two players are required to imitate each other, create synchronized and interesting motions and enjoy playing together, without any designation of leader and follower roles.
\end{enumerate}

Human movements in solo condition reflect their intrinsic dynamics, i.e., their individual motor signature \cite{piotr15}. On the other hand, participants reconcile their respective intrinsic dynamics with the communal goal (movement synchronization) in leader-follower or joint-improvisation condition. Here, we focus on the mathematical modeling of human coordination in solo and joint improvisation (JI) condition, and shed light on their interconnection.

\subsection{Motor signature}

Data analysis of experimental recordings reveals the self-similarity characteristics of human hand movements in solo trials, thus allowing to identify and distinguish human participants by comparing the kinematic features of their solo motions \cite{noy14,piotr14}. Indeed, motor signatures refer to the unique, time-persistent kinematic characteristics of human movements in solo condition \cite{piotr15,piotr14}. It has been shown that a possible candidate of motor signature is the probability distribution function (PDF) of velocity time series in solo trials \cite{piotr14}. As a consequence, a control architecture based on pre-recorded HP velocity profiles was developed for the VP to achieve real-time interaction in leader-follower and joint-improvisation conditions \cite{chao_cdc15,chao_ji15,chao_cdc14}.

Notably, skewness and kurtosis of normalized velocity segments provide also a suitable complement as marker of motor signature \cite{noy14}. Specifically, segments represent periods and portions of motion between two consecutive events of zero velocity, while normalized (or base) segments are obtained by normalizing the original ones over the time interval $[0,1]$ and the corresponding velocity integral. Figure~\ref{sg} gives a graphical representation of velocity-segments-based individual motor signatures, represented by the following ellipse:

\begin{equation}\label{ellip}
\frac{(z_s-\mu_s)^2}{\sigma^2_s}+\frac{(z_k-\mu_k)^2}{\sigma^2_k}=1
\end{equation}
where $z_s$ and $z_k$ represent the horizontal and vertical coordinates in the skewness-kurtosis (S-K) plane, with $\mu_s$ and $\mu_k$ ($\sigma_s$ and $\sigma_k$) referring to mean values (standard deviations) of skewness and kurtosis of the normalized velocity segments, respectively.

\begin{figure}
\scalebox{0.8}[0.8]{\includegraphics{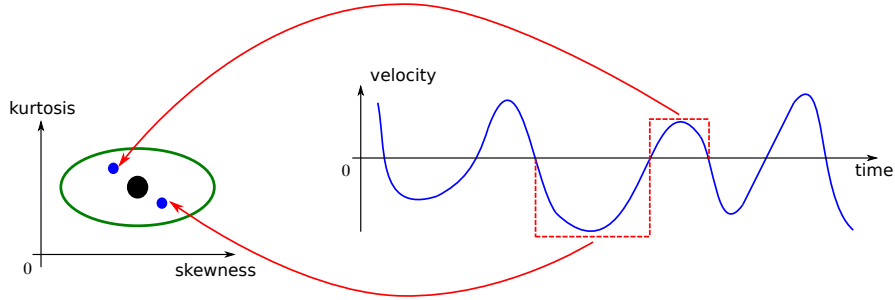}}\centering
\caption{\label{sg} Motor signature of a human participant based on velocity segments in the mirror game \cite{noy14}. The blue curve denotes the velocity time series of a human participant in a solo trial. The velocity segments in the red dashed boxes are normalized and then mapped as two blue points in the skewness-kurtosis (S-K) plane. Solo motions of a human participant in the S-K plane correspond then to a green ellipse, whose center is individuated by a black circle, which contains all the mapped segments.}
\end{figure}

Our goal is to develop a computational architecture for the VP to produce human-like solo movements and joint improvised trajectories with any desired values for skewness and kurtosis of normalized velocity segments, such that the kinematic features of a certain HP can be reproduced without making use of limited pre-recorded trajectories.

\subsection{Base segment of velocity}

It has been demonstrated that smooth point-to-point movements can be generated by minimizing the time integral of the jerk magnitude squared \cite{flash85}. This can be formulated as the following minimization problem:
\begin{equation}\label{min}
\min_{x}J(x)
\end{equation}
where
$$
J(x)=\frac{1}{2}\int_{0}^{1}\left(\frac{d^3x}{dt^3}\right)^2dt
$$
with $x(t), t\in [0, 1]$ denoting a desired position trajectory. In order to solve the optimization problem \eqref{min}, we first compute
\begin{equation}
J(x+c\delta x)=\frac{1}{2}\int_{0}^{1}\left(\frac{d^3x}{dt^3}+c \frac{d^3\delta x}{dt^3} \right)^2dt
\end{equation}
where $c$ is a constant and $\delta x(t), t\in [0, 1]$ is a smooth curve with the constraints
\begin{equation}\label{ini}
\delta x(0)=\frac{d^2\delta x(0)}{dt^2}=\frac{d^3\delta x(0)}{dt^3}=0
\end{equation}
and
\begin{equation}\label{fin}
\delta x(1)=\frac{d^2\delta x(1)}{dt^2}=\frac{d^3\delta x(1)}{dt^3}=0
\end{equation}
We then obtain the increment of $J(x)$
\begin{equation}
J(x+c\delta x)-J(x)=\frac{c}{2}\int_{0}^{1}\frac{d^3\delta x}{dt^3}\left(2\frac{d^3x}{dt^3}+c \frac{d^3\delta x}{dt^3}\right)dt
\end{equation}
that leads to
\begin{equation}
\lim_{c\rightarrow0}\frac{J(x+c\delta x)-J(x)}{c}=\int_{0}^{1}\frac{d^3\delta x}{dt^3}\cdot\frac{d^3x}{dt^3}dt
\end{equation}

From Equations \eqref{ini} and \eqref{fin} it follows that
\begin{equation}
\int_{0}^{1}\frac{d^3\delta x}{dt^3}\cdot\frac{d^3x}{dt^3}dt=-\int_{0}^{1}\delta x \cdot \frac{d^6x}{dt^6}dt
\end{equation}
The optimal trajectory should then satisfy
\begin{equation}\label{incrJ}
\lim_{c\rightarrow0}\frac{J(x+c\delta x)-J(x)}{c}=-\int_{0}^{1}\delta x \cdot \frac{d^6x}{dt^6}dt=0
\end{equation}
Since $\delta x$ can be an arbitrary function with initial condition (\ref{ini}) and terminal condition (\ref{fin}), Equation \eqref{incrJ} leads to a sixth-order differential equation
\begin{equation}\label{6th}
\frac{d^6x}{dt^6}=0
\end{equation}
Thus, an ideal solution to Equation \eqref{6th} is given by a fifth-order polynomial in $t$
\begin{equation}
x(t)=\sum_{i=0}^{5}a_it^i, \quad t\in [0,1]
\end{equation}
where $a_i, i\in \{0,1,2,3,4,5\}$ represent unknown coefficients. Therefore, the desired velocity segments correspond to a fourth-order polynomial in $t$.

In order to create a base segment of velocity that combines smooth motion with the desired kinematic features described by some individual motor signature, we define a probability distribution function
\begin{equation}
f(t):=\sum_{i=0}^{4} b_it^i, \quad t\in [0,1]
\end{equation}
where $b_i, i\in \{0,1,2,3,4\}$ represent unknown coefficients, and with the following boundary conditions
\begin{equation}\label{term}
f(0)=f(1)=0
\end{equation}
Mean value $\mu$ and variance $\sigma^2$ of $f(t)$ are defined as follows:
\begin{equation}\label{mu_sig}
\mu:=\int_{0}^{1}\tau f(\tau)d\tau, \quad \sigma^2:=\int_{0}^{1}(\tau-\mu)^2f(\tau)d\tau
\end{equation}

Since the integral of $f(t)$ over the time interval $[0,1]$ (i.e., the area of the base segment) must be unitary, that is
\begin{equation}\label{area}
\int_{0}^{1}f(\tau)d\tau=1
\end{equation}
Equations \eqref{term}, \eqref{mu_sig} and \eqref{area} yield $b_0=0$ and the following matrix equation
\begin{equation}\label{bmat}
\left(
  \begin{array}{cccc}
    1 & 1 & 1 & 1 \\
    \frac{1}{2} & \frac{1}{3} & \frac{1}{4} & \frac{1}{5} \\
    \frac{1}{3} & \frac{1}{4} & \frac{1}{5} & \frac{1}{6} \\
    \frac{1}{4} & \frac{1}{5} & \frac{1}{6} & \frac{1}{7} \\
  \end{array}
\right)\mathbf{b}=\left(
                    \begin{array}{c}
                      0 \\
                      1 \\
                      \mu \\
                      \mu^2+\sigma^2 \\
                    \end{array}
                  \right)
\end{equation}
where $\mathbf{b}=(b_1,b_2,b_3,b_4)^T$.
Likewise, the definitions of skewness $s$ and kurtosis $k$
\begin{equation}
s:=\frac{1}{\sigma^3}\int_{0}^{1}(\tau-\mu)^3f(\tau)d\tau, \quad k:=\frac{1}{\sigma^4}\int_{0}^{1}(\tau-\mu)^4f(\tau)d\tau
\end{equation}
are respectively equivalent to
\begin{equation}\label{s_sig}
\mathbf{b}^T\left(
              \begin{array}{c}
               \frac{1}{5}-\frac{3\mu}{4}+\frac{2\mu^2}{3} \\
               \frac{1}{6}-\frac{3\mu}{5}+\frac{\mu^2}{2} \\
               \frac{1}{7}-\frac{\mu}{2}+\frac{2\mu^2}{5} \\
               \frac{1}{8}-\frac{3\mu}{7}+\frac{\mu^2}{3} \\
              \end{array}
            \right)=s\sigma^3
\end{equation}
and
\begin{equation}\label{k_sig}
\mathbf{b}^T\left(
              \begin{array}{c}
               \frac{1}{6}-\frac{4\mu}{5}+\frac{3\mu^2}{2}-\mu^3 \\
               \frac{1}{7}-\frac{2\mu}{3}+\frac{6\mu^2}{5}-\frac{3\mu^3}{4} \\
               \frac{1}{8}-\frac{4\mu}{7}+\mu^2-\frac{3\mu^3}{5} \\
               \frac{1}{9}-\frac{\mu}{2}+\frac{6\mu^2}{7}-\frac{\mu^3}{2} \\
              \end{array}
            \right)=k\sigma^4
\end{equation}

By substituting $\bf{b}$ in Equations \eqref{s_sig} and \eqref{k_sig} with the solution to Equation \eqref{bmat}, we obtain a fourth-order polynomial system with two variables ($\mu$ and $\sigma$) and two parameters ($s$ and $k$) as follows
\begin{equation}\label{sys}
\left\{
  \begin{array}{ll}
    \mathcal{F}(\mu,\sigma,s)=0\\
    \mathcal{G}(\mu,\sigma,k)=0
  \end{array}
\right.
\end{equation}
where $\mathcal{F}(\mu,\sigma,s)=0$ and $\mathcal{G}(\mu,\sigma,k)=0$ correspond to (\ref{s_sig}) and (\ref{k_sig}), respectively.
The following result holds for the solution to Equation (\ref{sys}).
\begin{prop}\label{real_sol}
There exist real solutions $\mu$ and $\sigma$ to the polynomial system \eqref{sys} for any given positive parameters $s$ and $k$ characterizing the motor signature of a human player.
\end{prop}

\begin{proof}
See Appendix.
\end{proof}

\begin{remark}
Proposition \ref{real_sol} guarantees the existence of velocity segments satisfying smooth point-to-point movements with specified skewness and kurtosis. It is possible to prove Proposition \ref{real_sol} with the aid of discriminant \cite{kali_02} and resultant \cite{boch_07}.
\end{remark}

Analytical solutions to the polynomial system \eqref{sys} are not always available, hence numerical methods ($i.e.$, polynomial continuation) have to be used to find approximate solutions of mean value $\mu$ and standard deviation $\sigma$ for given skewness $s$ and kurtosis $k$. By means of approximated values of mean $\mu$ and standard deviation $\sigma$, it is possible to obtain the coefficient vector $\mathbf{b}=(b_1,b_2,b_3,b_4)^T$ and the base segment of velocity $f(t)=\sum_{i=0}^{4}b_it^i$ via Equation (\ref{bmat}).

For the sake of computational simplicity, in this work we assign all the four parameters $\mu$, $\sigma$, $s$ and $k$ characterizing the desired PDF $\mathcal{P}$ of a given HP, and then select three distinct time instants ($t_1$, $t_2$, $t_3$) for the fitted segment of velocity $h(t):=\sum_{i=0}^{4}c_it^i$ to match such velocity profile

\begin{equation}\label{fit}
h(t_i)=\mathcal{P}(t_i,\mu,\sigma,s,k), \quad t_i \in (0,1) \quad i \in \{1, 2, 3\}
\end{equation}
with
\begin{equation}\label{hterm}
h(0)=h(1)=0
\end{equation}

By combining Equations \eqref{fit} and \eqref{hterm}, we obtain the matrix equation
\begin{equation}\label{tildeb}
\left(
  \begin{array}{cccc}
    1 & 1 & 1 & 1 \\
    t_1 & t_1^2 & t_1^3 & t_1^4 \\
    t_2 & t_2^2 & t_2^3 & t_2^4 \\
    t_3 & t_3^2 & t_3^3 & t_3^4 \\
  \end{array}
\right)\mathbf{c}=\left(
                    \begin{array}{c}
                      0 \\
                      \mathcal{P}(t_1,\mu,\sigma,s,k)\\
                      \mathcal{P}(t_2,\mu,\sigma,s,k)\\
                      \mathcal{P}(t_3,\mu,\sigma,s,k)\\
                    \end{array}
                  \right)
\end{equation}
with $\mathbf{c}=(c_1,c_2,c_3,c_4)^T$. The solution to Equation \eqref{tildeb} gives the fitted segment of velocity
\begin{equation}\label{baseNotNorm}
h(t)=c_1t+c_2t^2+c_3t^3+c_4t^4, \quad t \in [0,1]
\end{equation}
which can finally be normalized to yield the fitted base segment of velocity
\begin{equation}\label{base}
g(t)=\frac{h(t)}{\int_{0}^{1}h(\tau)d\tau}
\end{equation}

\begin{figure}
\scalebox{0.55}[0.55]{\includegraphics{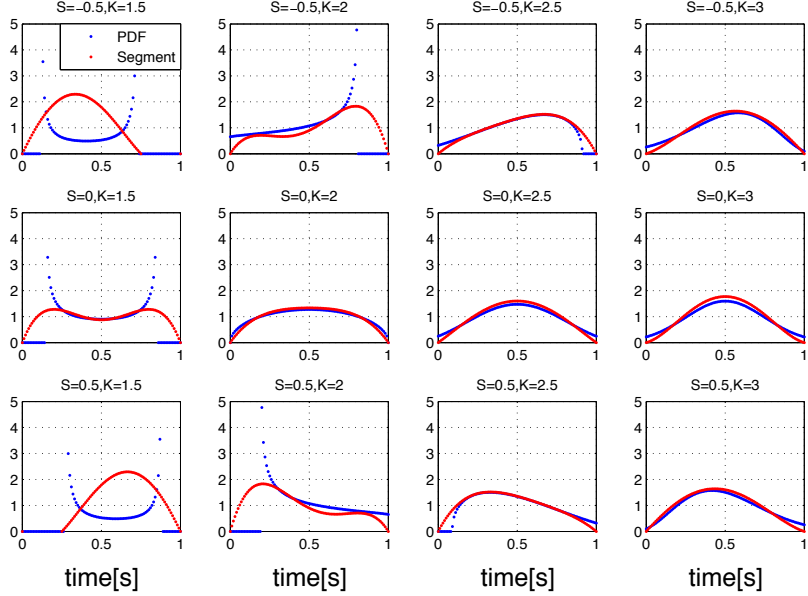}}\centering
\caption{\label{SK} Construction of fitted base segments of velocity by matching a desired velocity profile. The blue curve refers to the desired PDF $\mathcal{P}$ with specified skewness and kurtosis (mean value $\mu=0.5$ and standard deviation $\sigma=0.25$ are the same in all the sub-figures), while the red one represents the fitted base segment $g$. $S$ and $K$ stand for skewness and kurtosis, respectively. The values of skewness for human participants generally range between $-0.5$ and $0.5$, in comparison with those of kurtosis varying from $1.5$ to $3$, respectively \cite{noy14}.}\centering
\end{figure}

Figure~\ref{SK} presents twelve fitted base segments of velocity obtained for different values of skewness and kurtosis.

\section{Computational Architecture}\label{sec:ca}
The \emph{in-silico} generation of velocity trajectories in solo motion with prescribed kinematic features allows to develop a customized VP able to interact with a HP in JI condition, with the former exhibiting the desired motor signature of a given human participant. In this section we present the computational architecture of the VP to shed light on the relationship between the mechanism underlying the generation of solo and joint improvised motions. Compared with previous approaches \cite{chao_cdc15,chao_ji15,chao_cdc14}, the one we propose here allows the virtual player to spontaneously reproduce the motor signature of a given HP, without making use of pre-recorded time series of her/his motion in solo condition. This overcomes the drawback given by the need for a large database of human solo trajectories, and endows the VP with a wider repertoire of motor signatures, thus opening the possibility of exploring the effects of continuously changing its kinematic features during the interaction with another partner.

\begin{figure}
\scalebox{0.8}[0.8]{\includegraphics{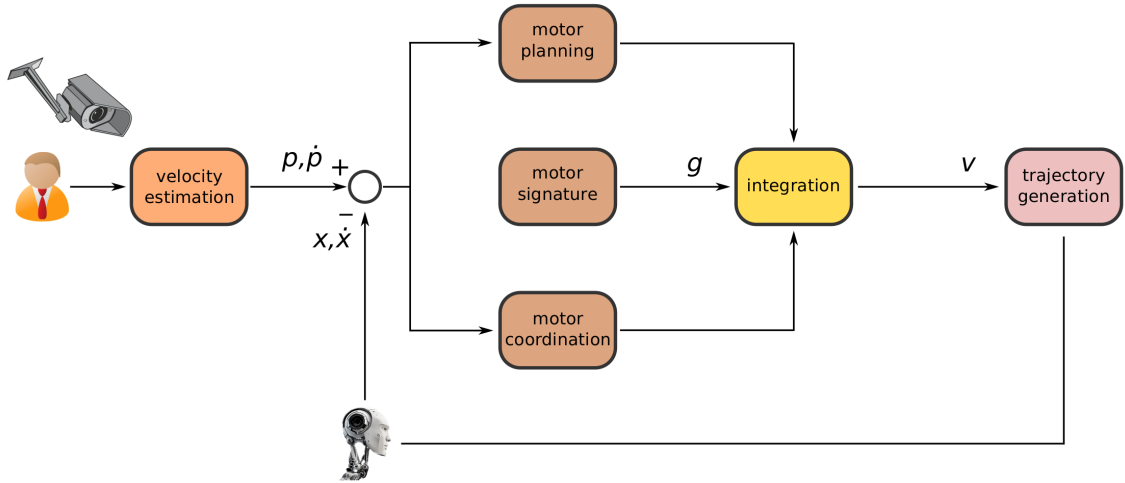}}\centering
\caption{\label{ca}Computational architecture of the VP in the mirror game. Variables $p$ and $\dot{p}$ represent position and velocity of the human player, while $x$ and $\dot{x}$ those of the virtual player; $g$ represents the fitted base segment and $v$ the actual velocity segment of the VP, respectively.}
\end{figure}

The proposed computational architecture (shown in Figure~\ref{ca}) consists of six function blocks described in details as follows.
\begin{enumerate}
  \item Velocity Estimation: The position trajectory of a HP detected by a camera is sent to this block, where her/his corresponding velocity time series is estimated and split into a series of velocity segments \cite{noy14}. Then position and velocity errors between HP and VP are computed.
  \item Motor Planning: This block determines the direction, duration and displacement of the velocity segments for the VP.
  \item Motor Signature: This block reflects the kinematic features of a human player as it generates the fitted base segment $g$. It allows to change the motor signature of the VP by resetting the desired values of $\mu$, $\sigma$, $s$ and $k$.
  \item Motor Coordination: This block allows for mutual adaptation, imitation and synchronization between the virtual player and its partner in joint improvisation condition.
  \item Movement Integration: The actual velocity segments $v$ of the VP are generated by integrating the movement constraints on motor planning, motor signature and motor coordination.
  \item Trajectory Generation: The movement trajectory of the VP is generated by chronologically assembling the integrated velocity segments.
\end{enumerate}

\subsection{Generation of solo motions}\label{sec:solo}
While playing the mirror game in solo condition, the VP produces a prescribed motion without taking into consideration that of any other participant. Thus, the generation of solo motions can be regarded as a special case of joint motion where there is no motor coordination.
Specifically, the actual segments of velocity $v$ are derived from the the fitted base segments $g$ after integrating the displacement with the duration of time, and after assigning a motion direction.

Let $\Delta t$ denote the duration of the time interval for each velocity segment, which is a random variable with probability distribution function $\lambda(\tau)$ that can be obtained by statistically analyzing the solo recordings of a human participant. The probability of $\Delta t$ belonging to the interval $[\underline{t},\bar{t}]$ can be calculated as
\begin{equation}\label{Dt}
P\left(\underline{t}\leq \Delta t\leq \bar{t}\right)=\int_{\underline{t}}^{\bar{t}}\lambda(\tau)d\tau
\end{equation}
According to experimental data, the average time interval for velocity segments is equal to $0.8$s, with a standard deviation of $0.7$s \cite{noy14}. In addition, let $\Delta l$ represent the segment displacement (i.e., position mismatch between the starting point and terminal point of each segment), which is a random variable with probability distribution function
$\xi(s)$. Likewise, the probability of $\Delta l$ belonging to the interval $[\underline{l},\bar{l}]$ is given by
 \begin{equation}\label{Dl}
P\left(\underline{l}\leq \Delta l\leq \bar{l}\right)=\int_{\underline{l}}^{\bar{l}}\xi(s)ds
\end{equation}
Regardless of the motion direction, the variant of a fitted base segment can be calculated as
\begin{equation}\label{variantEquationNoDir}
\frac{\Delta l}{\Delta t}\cdot g\left(\frac{t}{\Delta t}\right)
\end{equation}
where $g$ is defined in Equation \eqref{base}. Figure~\ref{variant} shows a fitted base segment of velocity and possible eight variants for it with respect to time duration $\Delta t$ and displacement $\Delta l$.

\begin{figure}
\scalebox{0.6}[0.6]{\includegraphics{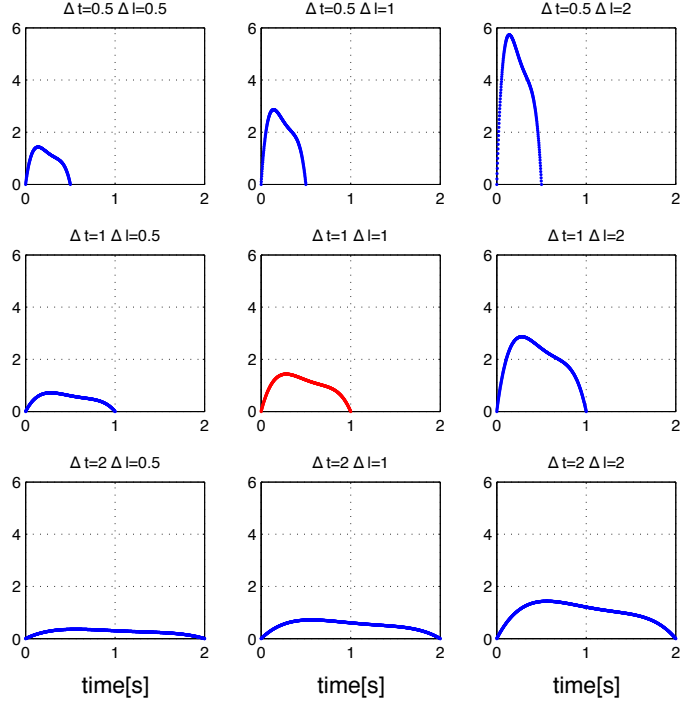}}\centering
\caption{\label{variant} Variants of a fitted base segment of velocity with respect to time duration $\Delta t$ and displacement $\Delta l$. The red curve represents $g$, while the blue ones represent its variants obtnained for different values of $\Delta t$ and $\Delta l$ as described in Equation \eqref{variantEquationNoDir}.}
\end{figure}

Since HPs tend to move around the middle part of the string in solo trials \cite{piotr15}, the movement direction of the VP is determined by
\begin{equation}\label{dire}
\vec{D}=\left\{
       \begin{array}{ll}
         \mbox{sign}(x-p_b), & \hbox{$|x-p_a|>|x-p_b|$;} \\
         \mbox{sign}(x-p_a), & \hbox{$|x-p_a|<|x-p_b|$;} \\
         \mbox{either}, & \hbox{$|x-p_a|=|x-p_b|$,}
       \end{array}
     \right.
\end{equation}
where $x$ denotes the position of the VP, and $p_a<p_b$ represent position bounds. An actual velocity segment $v$ is then constructed as follows
\begin{equation}\label{var}
v(t)=\vec{D}\cdot\frac{\Delta l}{\Delta t}\cdot g\left(\frac{t}{\Delta t}\right)=\vec{D}\cdot\frac{\Delta l\cdot h(\frac{t}{\Delta t})}{\Delta t \cdot \int_{0}^{1}h(\tau)d\tau} \quad t\in [0, \Delta t]
\end{equation}

Solo motions are generated by consecutively joining the actual velocity segments together. Finally, the position trajectory of the VP is produced as follows
\begin{equation}\label{pos}
x(t)=x_0+\int_{0}^{t}v(\tau)d\tau \quad t\in [0, \Delta t]
\end{equation}
where $x_0$ denotes the initial position of the generated segment. Table~\ref{table} summarizes the solo motion algorithm (SMA) employed for the VP to produce human-like solo movements with prescribed kinematic features.

\begin{table}
 \caption{\label{table} Solo Motion Algorithm (SMA).}
 \begin{center}
 \begin{tabular}{lcl} \hline
  1: Set skewness $s$, kurtosis $k$ and running time $T_s$ \\
  2: Generate a fitted base segment $g(t)$ with \eqref{fit}, \eqref{hterm}, \eqref{tildeb}, \eqref{baseNotNorm} and \eqref{base} \\
  3: \textbf{while} ($\mbox{time}<T_s$) \\
  4: ~~~~~~~Determine the segment duration $\Delta t$ with (\ref{Dt}) \\
  5: ~~~~~~~Determine the segment displacement $\Delta l$ with (\ref{Dl}) \\
  6: ~~~~~~~Choose the movement direction $\vec{D}$ with (\ref{dire}) \\
  7: ~~~~~~~Generate an actual velocity segment $v(t)$ with (\ref{var}) \\
  8: ~~~~~~~Output the position trajectory $x(t)$ with (\ref{pos}) \\
  9: \textbf{end while} \\ \hline
 \end{tabular}
 \end{center}
\end{table}

\subsection{Generation of joint improvised motions}\label{sec:joint}


While playing the mirror game in JI condition, the VP interacts with its partner while exhibiting some prescribed kinematic features (motor signature). Based on the position and velocity mismatch between the two players, the proposed computational architecture allows the virtual player to imitate, adapt to and synchronize with the movement of its partner, thereby achieving joint improvisation \cite{chao_ji15}.

Similarly to SMA, the segment duration and displacement are determined by Equations \eqref{Dt} and \eqref{Dl}, respectively. As the two participants attempt to achieve movement synchronization, the movement direction
of the VP is given by
\begin{equation}\label{dire_ji}
\vec{D}=\mbox{sign}(p-x)
\end{equation}
where $x$ denotes the position of the virtual player and $p$ refers to that of the other agent. When $p=x$, the VP is provided with a random direction.

The motor coordination block enables the VP to imitate and adapt to the movement of its partner in order to synchronize their joint movements, while the two participants consciously adjust their way of moving (i.e., the profile of their velocity segments during the game).
It has been suggested that an optimal feedback control driving the VP is equivalent to a PD control when the optimization interval is small enough, and that the nonlinear HKB equation originally introduced in \cite{hkb85} is not significantly better than a double integrator as end effector model of the VP in the mirror game \cite{fran_ecc16}.

For the sake of simplicity, in this work we employ a double integrator with PD control to describe the motion of the VP and design the online algorithm as follows
\begin{equation}\label{learn}
\ddot{x}=c_s(v-\dot{x})+c_v(\dot{p}-\dot{x})+c_p(p-x)+\kappa(x,\epsilon)
\end{equation}
where $v$ is the actual velocity segment generated by Equation \eqref{var}, $x$ and $\dot{x}$ represent position and velocity of the VP, $p$ and $\dot{p}$ those of its partner, with $c_s$, $c_v$, $c_p$ and $k$ being tunable positive parameters. The first three terms on the right-hand side of Equation \eqref{learn} account for preferred movement, mutual imitation and movement synchronization, respectively \cite{chao_ji15}, whereas $\kappa(x,\epsilon)$ is used to constrain the movement of the VP within the admissible range of motion:
$$
\kappa(x,\epsilon)=\left\{
       \begin{array}{ll}
         c_r|x-p_b|, & \hbox{$x-p_a\leq\epsilon$} \\
        -c_r|x-p_a|, & \hbox{$p_b-x\leq\epsilon$} \\
                0, & \hbox{otherwise}
       \end{array}
     \right.
$$
with $c_r$ and $\epsilon$ being tunable positive parameters.
When the distance between the VP and its closer bound is lower than $\epsilon$, the term $\kappa(x,\epsilon)$ drives the VP with strength $c_r$ towards the middle point of the position range.

By solving equation (\ref{learn}), the position trajectory of the VP is given by
\begin{equation}\label{posj}
x(t)=x_0+\int_{0}^{t}\int_{0}^{\tau}\ddot{x}(s) ds \ d\tau,  \quad t\geq 0
\end{equation}
where $x_0$ refers to the initial position of the VP. Table~\ref{table_ji} summarizes the joint improvisation algorithm (JIA) employed for the VP to perform JI with another agent in the mirror game.

\begin{table}
 \caption{\label{table_ji} Joint Improvisation Algorithm (JIA).}
 \begin{center}
 \begin{tabular}{lcl} \hline
  1: Set skewness $s$, kurtosis $k$ and running time $T_s$ \\
  2: Generate a fitted base segment $g(t)$ with \eqref{fit}, \eqref{hterm}, \eqref{tildeb}, \eqref{baseNotNorm} and \eqref{base} \\
  3: \textbf{while} ($\mbox{time}<T_s$) \\
  4: ~~~~~~~Determine the segment duration $\Delta t$ with (\ref{Dt}) \\
  5: ~~~~~~~Determine the segment displacement $\Delta l$ with (\ref{Dl}) \\
  6: ~~~~~~~Choose the movement direction $\vec{D}$ with (\ref{dire_ji}) \\
  7: ~~~~~~~Generate an actual velocity segment $v(t)$ with (\ref{var}) \\
  8: ~~~~~~~Evaluate the acceleration $\ddot{x}(t)$ with (\ref{learn})  \\
  9: ~~~~~~~Output the position trajectory $x(t)$ with (\ref{posj}) \\
  10: \textbf{end while} \\ \hline
 \end{tabular}
 \end{center}
\end{table}

\section{Experimental Validation}\label{sec:exp}

In order to test and validate the proposed computational architecture, in this section we compare solo and joint improvised motions of human players with those generated by their respective customized virtual agents. The numerical algorithms are implemented in Matlab R2010a.

\subsection{Solo motions}

Figures~\ref{solo_sim_fig}(a) and \ref{solo_sim_fig}(b) show position and velocity time series of a HP performing a $60$s solo trial. The HP moves the ball along the string within the normalized range $[-1,1]$. The sampling frequency of the camera is $100$ Hz. According to data analysis of the velocity segments shown in Fig.~\ref{solo_sim_fig}(b), the averaged mean value $\bar{\mu}$, standard deviation $\bar{\sigma}$, skewness $\bar{s}$ and kurtosis $\bar{k}$ are $0.50$, $0.23$, $-0.08$ and $2.11$, respectively. We then choose three time points $t_1=\bar{\mu}-\bar{\sigma}$, $t_2=\bar{\mu}$ and $t_3=\bar{\mu}+\bar{\sigma}$ to construct the base segment of velocity. In particular, the Matlab function ``pearspdf'' is employed to compute the values of the desired PDF $\mathcal{P}(t,\bar{\mu},\bar{\sigma},\bar{s},\bar{k})$ at the selected time points.

\begin{figure}[!h]
\scalebox{0.9}[0.9]{\includegraphics{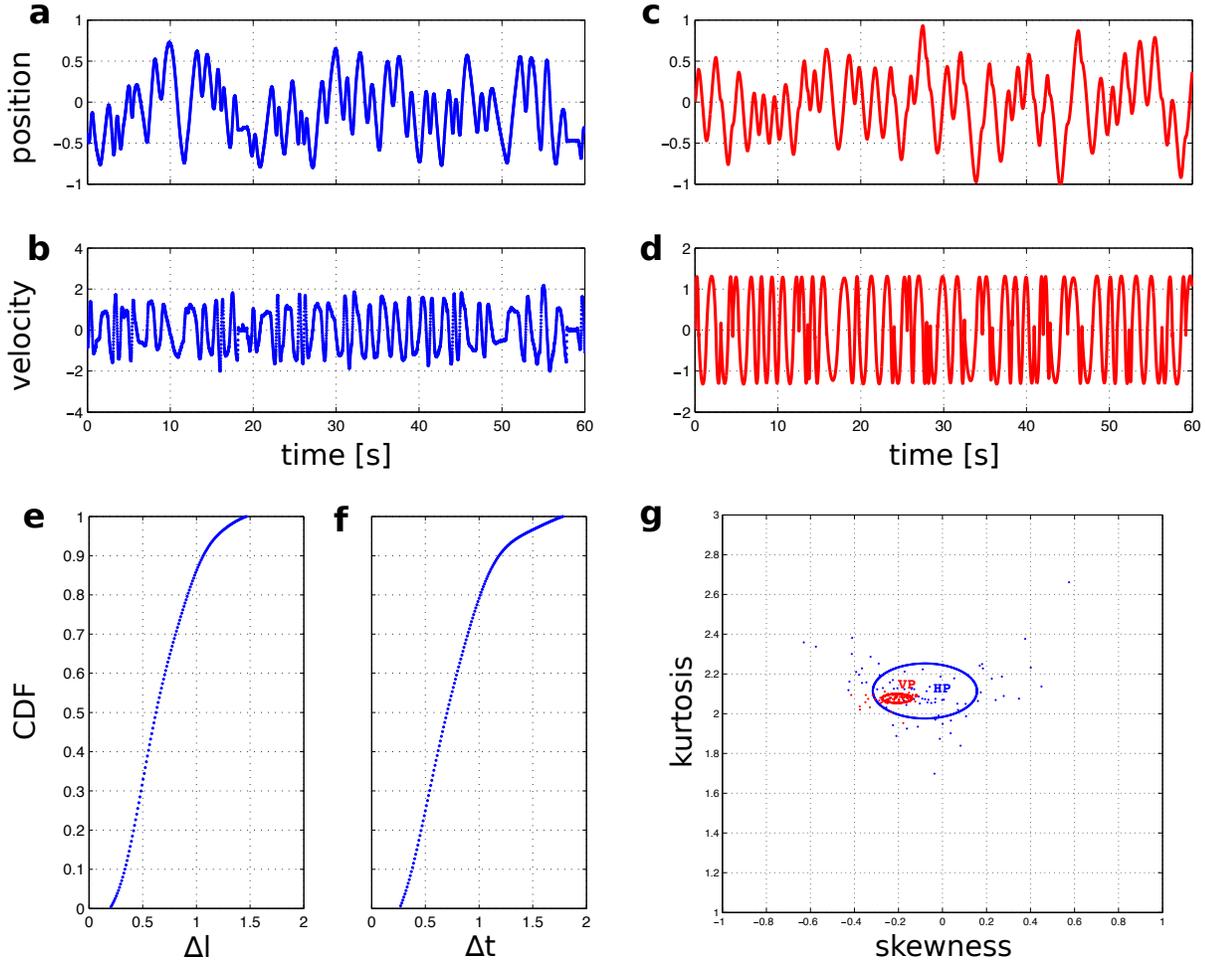}}\centering
\caption{\label{solo_sim_fig} Experimental validation -- solo motions. Position (a) and velocity (b) time series of the HP. Position (c) and velocity (d) time series of the VP. CDFs of $\Delta l$ (e) and $\Delta t$ (f) for the HP. (g) Visualization of solo motion for the HP and her/his customized VP in the S-K plane: blue dots correspond to velocity segments of the HP, whereas red ones refer to those of the VP. The two corresponding ellipses are evaluated by means of Equation \eqref{ellip}.}
\end{figure}

The probability distributions of $\Delta l$ and $\Delta t$ of the velocity segments in Fig.~\ref{solo_sim_fig}(b) are described by cumulative distribution functions (CDF) shown in Figs.~\ref{solo_sim_fig}(e) and \ref{solo_sim_fig}(f), respectively.

Figures~\ref{solo_sim_fig}(c) and \ref{solo_sim_fig}(d) show position and velocity time series of a VP fed with the same motor signature as that in Fig.~\ref{solo_sim_fig}(b) and driven by the SMA described in Table~\ref{table}. The velocity segments generated by the SMA resemble those of the HP in terms of profile, yet are slightly smoother. A visible difference is that the HP sometimes stays still during the game, whilst the VP always keeps moving.

Figure~\ref{solo_sim_fig}(g) shows skewness and kurtosis of normalized velocity segments for both the HP and her/his customized VP in the S-K plane.
It is possible to appreciate that most velocity segments of the VP are mapped into the ellipse representing the kinematic features of the HP, thus confirming hat the VP succeeds in reproducing the motor signature of the specified HP. Moreover, the VP segments are clustered together, whereas those of the HP are scattered in the S-K plane, thus implying that solo motions of human players are more flexible and diverse than those of their customized computer avatar.

\subsection{Joint improvised motions}
Next, we present numerical validation of the JIA described in Table~\ref{table_ji} for both HP-VP and VP-VP dyads in a joint improvisation condition.

\subsubsection{HP-VP dyad}

The experimental set-up allowing a HP to perform joint improvisation with a VP is shown in Fig.\ref{cahpvp}. The parameter setting for the VP is given as follows: $\bar{\mu}=0.51$, $\bar{\sigma}=0.23$, $\bar{s}=-0.09$, $\bar{k}=2.14$, $c_s=2$, $c_v=5$, $c_p=3$, $c_r=5$ and $\epsilon=0.1$.

\begin{figure}[!h]
\scalebox{0.65}[0.65]{\includegraphics{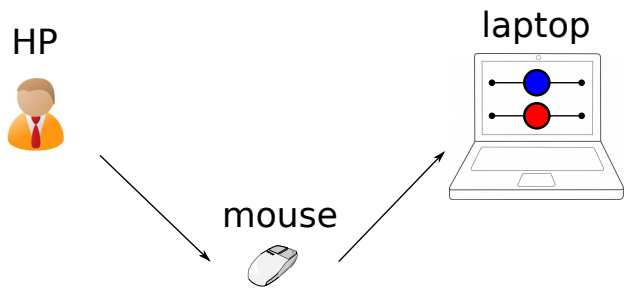}}\centering
\caption{\label{cahpvp}Experimental set-up of HP-VP interaction in the mirror game. The HP is required to sit in front of a laptop, which implements the JIA in Matlab. The blue circle represents the position of the HP, which is controlled by means of a mouse, while the red circle represents that of the VP, which is generated by the JIA. }
\end{figure}

Figures~\ref{jihpvp}(a) and \ref{jihpvp}(b) show position and velocity time series of HP and VP, respectively. Some synchronized segments can be observed in the position trajectories, which implies the occurrence of joint improvisation between HP and VP.

\begin{figure}[!h]
\scalebox{0.75}[0.75]{\includegraphics{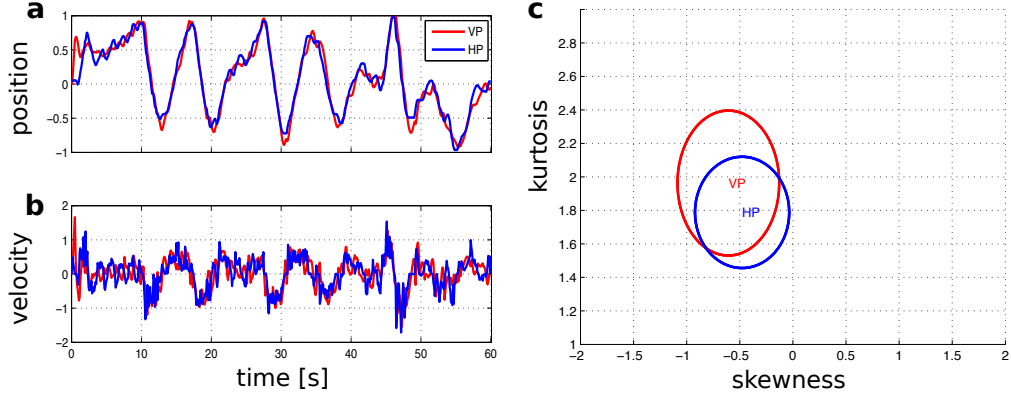}}\centering
\caption{\label{jihpvp} Experimental validation -- JI trial between HP (blue) and VP (red). Position (a) and velocity (b) time series of HP and VP. (c) Visualization of the JI motion between HP and VP in the S-K plane.}
\end{figure}

The two ellipses featuring the movement patterns of the two interacting agents are shown in Fig.~\ref{jihpvp}(c). It is possible to appreciate that they are largely overlapping in the S-K plane, implying that the two players exhibit similar kinematic features while interacting in the mirror game.

\subsubsection{VP-VP dyad}

In order to validate the capability of the proposed computational architecture to reproduce the kinematic characteristics observed when two human players (HP1 and HP2) perform the mirror game in a joint improvisation condition,
we numerically simulate a VP-VP trial. The evaluation method is the same as that proposed in \cite{chao_ji15}. Specifically, two virtual players (VP1 and VP2) are enabled to play the mirror game in a JI condition, with VP1 (VP2) being fed with the motor signatures of HP1 (HP2), respectively (Fig.~\ref{ca2vp}).

\begin{figure}[!h]
\scalebox{0.75}[0.75]{\includegraphics{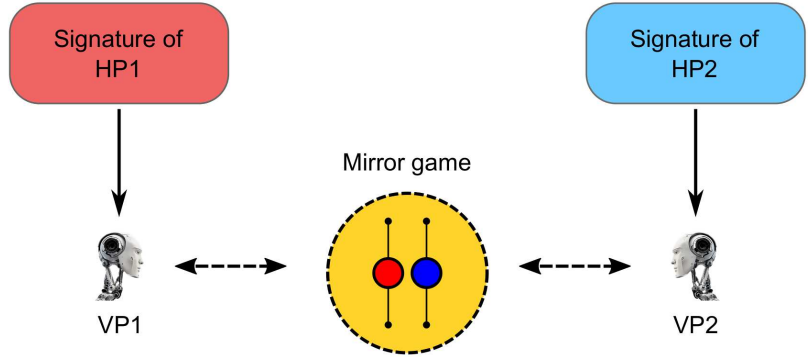}}\centering
\caption{\label{ca2vp}Schematic diagram of VP-VP interaction in the mirror game.}
\end{figure}

The two virtual players are driven by the JIA with the following parameters setting: $\bar{\mu}_1=0.51$, $\bar{\sigma}_1=0.22$, $\bar{s}_1=-0.18$ and $\bar{k}_1=2.13$ for VP1, $\bar{\mu}_2=0.53$, $\bar{\sigma}_2=0.25$, $\bar{s}_2=-0.18$ and $\bar{k}_2=1.87$ for VP2, and $c_s=1.5$, $c_v=3.6$, $c_p=4.9$, $c_r=5$ and $\epsilon=0.1$ for both VPs. Figures~\ref{jivpvp}(a) and \ref{jivpvp}(b) show position and velocity time series of the two human players, while Figures~\ref{jivpvp}(c) and \ref{jivpvp}(d) those of the two customized virtual agents, respectively.
VP1 and VP2 succeed in reproducing the joint improvised movement (synchronized segments) as occurred in the HP1-HP2 interaction.

\begin{figure}[!h]
\scalebox{0.85}[0.85]{\includegraphics{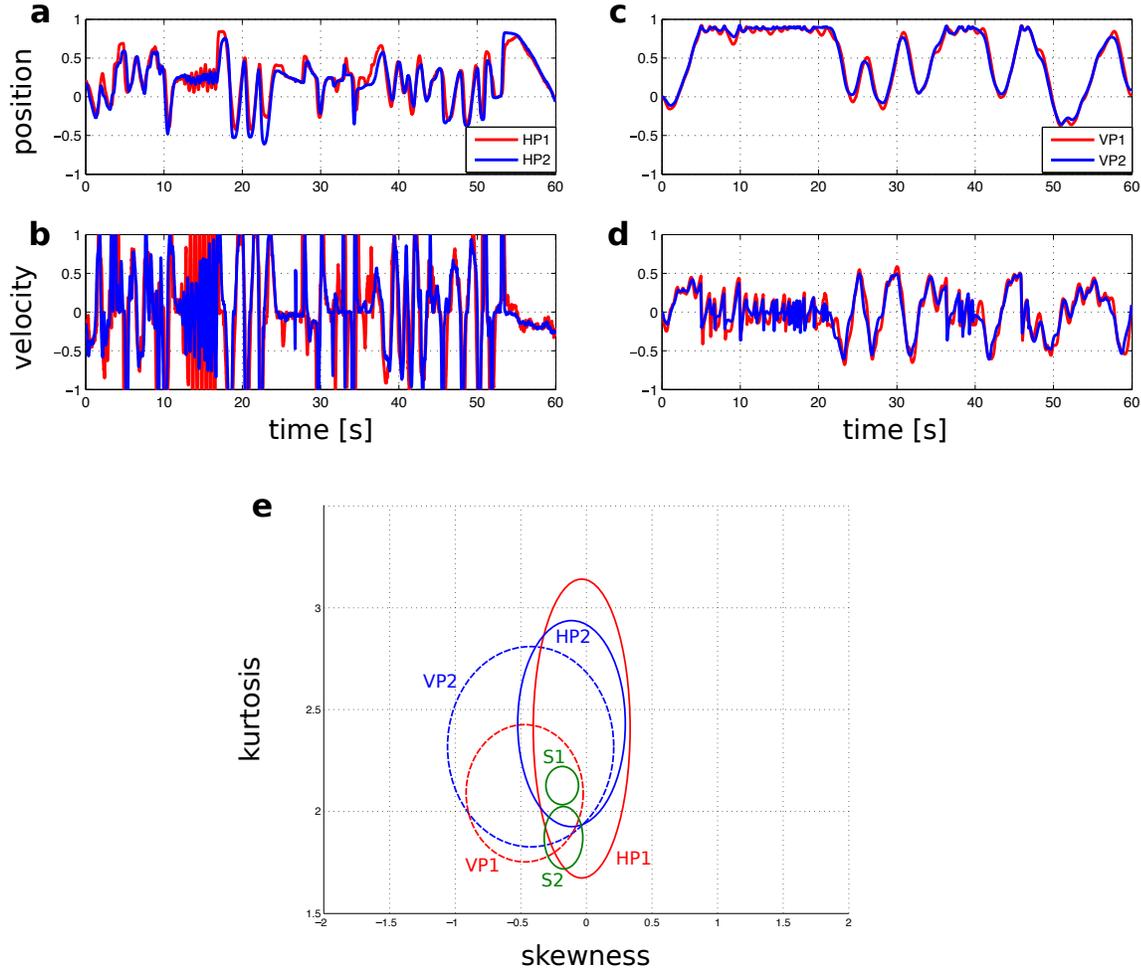}}\centering
\caption{\label{jivpvp} Experimental validation -- JI trial in a human (HP1 and HP2) and in a virtual (VP1 and VP2) dyad. Position (a) and velocity (b) time series of the human dyad (HP1 in red and HP2 in blue). Position (c) and velocity (d) time series of the virtual dyad (VP1 in red and VP2 in blue). (e) Visualization of solo and JI motions for the human pair and the customized virtual pair in the skewness-kurtosis plane. VP segments are mapped into dashed-line ellipses (VP1 in red and VP2 in blue), HP segments into solid-line ellipses (HP1 in red and HP2 in blue), and their corresponding kinematic signatures in solo motion (S1 and S2) into green solid-line ellipses.}
\end{figure}

Figure~\ref{jivpvp}(e) describes the transition of motor signatures from solo to JI motion. The kinematic features of the human players in solo condition are separate, while those in JI condition converge towards each other and are more variable. Notably, similar remarks can be made for the kinematic features exhibited by the virtual players, thus indicating the desirable matching performance of the VPs driven by the proposed computational architecture.

\section{Conclusions}\label{sec:con}
We developed a systematic approach to account for the generation of human solo motions, joint improvised motions and the transition of their kinematic characteristics in the mirror game. In so doing, a computational architecture was designed to describe the mechanisms underlying solo and joint improvised movements, which provides a new insight into the shift of kinematic patterns from individuality to joint improvisation.

We observed how, despite being characterized by different motor signatures in solo motion, players tend to imitate their respective kinematic features when interacting together, and exhibit a wider repertoire of movements. Such results were successfully captured by the proposed computational architecture, thus  opening the possibility of testing \emph{in-silico} interactions between different individuals in a number of different configurations.
Theoretical analysis was also presented to guarantee the existence of base segments of velocity characterizing any individual motor signature.

Future work may include the consideration of motor learning in joint actions and the generalization of this approach to other experimental paradigms for investigating socio-motor coordination, both in dyads \cite{fran16a} and in larger ensembles \cite{michael_tac15,michael_iet15,fran16b,fran16c}.

\section*{Acknowledgments}

The authors wish to thank Prof. Krasimira~Tsaneva-Atanasova and Dr. Piotr~S\l{}owi\'{n}ski at the University of Exeter, UK for the insightful discussions and thank Prof. Benoit Bardy, Prof. Ludovic Marin and Dr. Robin Salesse at the University of Montpellier, France for collecting the experimental data that is used to validate the approach presented in this paper. This work is supported by National Nature Science Foundation of China under Grant 61374053, by the Innovation and Technology Commission under Grant No. UIM/268, and by the Research Grants Council, Hong Kong, through the General Research Fund under Grant No. 17205414.

\section*{Appendix}

In what follows we present the details on the proof of Proposition \ref{real_sol}.

\begin{proof}
$\mathcal{F}(\mu,\sigma,s)$ and $\mathcal{G}(\mu,\sigma,k)$ in Equation \eqref{sys} can be simplified as follows:
\begin{equation}
\mathcal{F}(\mu,\sigma,s)=-\mu^3-3\mu\sigma^2-s\sigma^3+\frac{3}{2}\mu^2+\frac{3}{2}\sigma^2-\frac{9}{14}\mu+\frac{1}{14}
\end{equation}
and
\begin{equation}
\mathcal{G}(\mu,\sigma,k)=3\mu^4+6\mu^2\sigma^2-k\sigma^4-6\mu^3-6\mu\sigma^2+\frac{89}{21}\mu^2+\frac{5}{3}\sigma^2-\frac{26}{21}\mu+\frac{5}{42}
\end{equation}
which can be rewritten as
\begin{equation}
\mathcal{F}_1(\mu,\sigma,s):=\frac{\mathcal{F}(\mu,\sigma,s)}{\sigma^3}=-s+\left(\frac{3}{28\sigma^2}-3\right)\frac{\mu-1/2}{\sigma}- \left(\frac{\mu-1/2}{\sigma}\right)^3
\end{equation}
and
\begin{equation}
\mathcal{G}_1(\mu,\sigma,s):=\frac{\mathcal{G}(\mu,\sigma,k)}{\sigma^4}=
-k+\frac{1}{6\sigma^2}-\frac{1}{336\sigma^4}+\left(6-\frac{11}{42\sigma^2}\right)\left(\frac{\mu-1/2}{\sigma}\right)^2+
3\left(\frac{\mu-1/2}{\sigma}\right)^4 \, .
\end{equation}

From these representations, it is evident that if the system has a solution $ (\mu, \sigma) \in \mathbb C^2 $ then it also has a solution $ (1-\mu, -\sigma) $. Furthermore, with the aid of substitution
\begin{equation}\label{UTsubs1}
M=\frac{\mu-1/2}{\sigma},\ \eta=\frac{1}{\sigma^2}
\end{equation}
the expressions for $\mathcal{F}_1(\mu,\sigma,s)$ and $\mathcal{G}_1(\mu,\sigma,k)$ can be further simplified as
\begin{equation}
\mathcal{F}_1(M,\eta,s)=-s+\left(\frac{3}{28}\eta-3\right)M-M^3 \ ,
\end{equation}
and
\begin{equation}
\mathcal{G}_1(M,\eta,k)=-k+\frac{1}{6}\eta-\frac{1}{336}\eta^2 +\left(6-\frac{11}{42}\eta\right)M^2+3\,M^4,
\end{equation}
respectively. By solving equation $ \mathcal{F}_1(M,\eta,s)=0 $ with respect to $ \eta $, we obtain
\begin{equation}
\eta= \frac{28}{3}\left(M^2+3+\frac{s}{M}\right)
\label{UTeta}
\end{equation}
Substitution of Equation \eqref{UTeta} into $ \mathcal{G}_1(M,\eta,k)=0 $ yields
\begin{equation}
\mathcal{G}_2(M,s,k)=0
\label{UTfinal}
\end{equation}
with
\begin{equation}
\mathcal{G}_2(M,s,k)=8\,M^6-36\,M^4-80\,sM^3+ (63-27\, k)M^2-7\,s^2 \, .
\end{equation}

According to data analysis of human movements in the mirror game, skewness $s$ and kurtosis $k$ belong to the intervals $(-0.5, 0.5)$ and $(1.5, 3)$, respectively \cite{noy14}. For any selection of values $s\in (0, 0.5)$ and $k\in(1.5, 3)$, Equation \eqref{UTfinal} has a positive zero $M=M_0$ in the interval $(0,2\sqrt{3})$ due to the conditions
\begin{equation}
\mathcal{G}_2(0,s,k)<0,
\end{equation}
and
\begin{equation*}
\mathcal{G}_2(2\sqrt{3},s,k)=9396-1920\sqrt{3}s-324\,k-7\,s^2
\end{equation*}
\begin{equation}
=(9065-1920\sqrt{3}s)+324(1-k)+7(1-s^2)>0 \, .
\end{equation}

Therefore, from Equation \eqref{UTeta} it is clear that also $\eta$ is positive, hence the second equation from \eqref{UTsubs1} can be resolved in real numbers with respect to $ \sigma $. The corresponding value for $\mu$ can be then found in the first equation from (\ref{UTsubs1}), which implies that $\mathcal{F}(\mu,\sigma,s)=0$ and $\mathcal{G}(\mu,\sigma,k)=0$ have real roots $\mu$ and $\sigma$.
\end{proof}

\end{document}